\documentclass{article}
\usepackage{amsmath}
\usepackage{amssymb}
\usepackage{amsfonts}
\usepackage{a4wide}
\usepackage{amsthm}
\usepackage{tikz}
\usepackage{float}

\newcounter{NN}
\newtheorem{theorem}[NN]{Theorem}

\newtheorem{lemma}[NN]{Lemma}

\def\N{\mathbb{N}}
\def\A{\mathbf{A}}
\def\B{\mathbf{B}}
\def\Z{\mathbf{Z}}

\begin{document}

\title{Linear Darboux polynomials for Lotka-Volterra systems, trees and superintegrable families}

\author{
G.R.W. Quispel$^1$, Benjamin K. Tapley$^2$, D.I. McLaren$^1$, Peter H.~van der Kamp$^1$
\\[2mm]
$^1$Department of Mathematical and Physical Sciences,\\ La Trobe University, Victoria 3086, Australia.\\
$^2$Department of Mathematics and Cybernetics, SINTEF Digital, 0373 Oslo, Norway.\\[2mm]
Email: P.vanderKamp@LaTrobe.edu.au\\[7mm]
}
\maketitle

\begin{abstract}
We present a method to construct superintegrable $n$-component Lotka-Volterra systems with $3n-2$ parameters. We apply the method to Lotka-Volterra systems with $n$ components for $1<n<6$, and present several $n$-dimensional superintegrable families. The Lotka-Volterra systems are in one-to-one correspondence with trees on $n$ vertices.
\end{abstract}

\section{Introduction}
The original 2-dimensional Lotka-Volterra (LV) system,
\begin{equation} \label{OLV}
\dot{x}=x(a-by), \qquad \dot{y}=y(-c+dx),
\end{equation}
where $\dot{x}$ denotes the derivative with respect to time, was derived as a model to describe the interaction between predator and prey fish \cite{Lot,Vol,DP}. Sternberg \cite[Chapter 11]{Ste} gives a dynamical systems perspective and an explanation why fishing decreases the number of predators. The 2-dimensional system (\ref{OLV}) has been generalised to $n$-dimensional systems of the form
\begin{equation} \label{GLV}
\dot{x}_i = x_i \left(b_i + \sum_{i}^{} A_{i,j} x_j \right),
\end{equation}
where $\mathbf{b}$ is a real vector, and $\mathbf{A}$ is a real matrix, and these have been studied extensively. For references on various aspects of LV systems, including integrability as well as their history, see \cite{BBM,Bog,BV,BZDK,CD,Dam,DEKV,DP,EKV,HBF,Ito,KQV,Pla}. Prelle and Singer wrote a very influential paper \cite{PS} proving that if a polynomial ODE has an elementary integral, then it has a logarithmic integral. Note that in the mathematical physics literature the matrix $\A$ is often assumed to be skew symmetric. This is not assumed here.

A vector field on an $n$-dimensional manifold is called {\it superintegrable} if it admits $n-1$ functionally independent constants of motion (i.e. first integrals), cf. \cite{KKQTV}.  In this paper we construct superintegrable $n$-component Lotka-Volterra systems with $3n-2$ parameters.

Darboux polynomials (DPs) are building blocks of rational integrals and their generalizations \cite{Darb,Gor}. Given an ordinary differential equation (ODE)
\[
\frac{d\mathbf{x}}{dt}=\mathbf{f}(\mathbf{x}),
\]
where $\mathbf{x}(t)$ and $\mathbf{f}$ are $n$-dimensional vectors, a Darboux polynomial $P(\mathbf{x})$ is defined by the existence of a polynomial $C(\mathbf{x})$ s.t.
\begin{equation}\label{defDP}
\frac{dP(\mathbf{x})}{dt} = C(\mathbf{x})P(\mathbf{x})
\end{equation}
Note that (\ref{defDP}) implies that if $P(\mathbf{x}(0)) = 0$, then $P(\mathbf{x}(t)) = 0, \forall t$. For this reason Darboux polynomials are also called second integrals.

In section 2, we provide a method to obtain $m$ integrals for an $n$-dimensional homogeneous quadratic ODE, from $m+n$ Darboux polynomials.
In section 3, we give conditions on $\mathbf{b}$ and $\mathbf{A}$ which are equivalent to
\[
P_{i,k}=\alpha x_i +\beta x_k
\]
being a DP for \eqref{GLV}. In section 4, we look at the intersection of the above two classes, i.e. at homogeneous Lotka-Volterra systems, and use the described method and mentioned DPs to construct some superintegrable systems in dimensions 2, 3, and 4. In section 5, we explain how these superintegrable $n$-dimensional LV systems are in one-to-one correspondence with trees on $n$ vertices. Such a tree has $n-1$ edges, and each of these edges corresponds to an integral. If an edge exists between vertices $i$ and $k$, the corresponding integral can be written as a product of $P_{i,k}$ and powers of the variables $x_j$, $j=1\ldots n$.
In section 6, we cover the superintegrable LV-systems which relate to the 3 non-isomorphic trees on 5 vertices. We also describe the factorisation of the exponents of the variables in terms of minors of the matrix $\mathbf{A}$.
In our final section we give some details for the superintegrable $n$-dimensional LV systems that relate to tall trees. In the appendix we explain how the Euler top relates to a special case of our superintegrable 3-dimensional LV system.

\section{A rather general method} \label{rgm}
Let
\[
\frac{dP_1}{dt} = C_1P_1, \qquad
\frac{dP_2}{dt} = C_2P_2
\]
then
\[
\frac{d}{dt}\left(P_1^{\alpha_1}P_2^{\alpha_2} \right) = (\alpha_1C_1 + \alpha_2C_2) P_1^{\alpha_1}P_2^{\alpha_2}.
\]
Hence cofactors $C_i$ form a linear space. Note that $C_1=C_2$ if and only if $\frac{P_1}{P_2}$ is an integral. We also have
 \begin{equation}
 P_1^{\alpha_1}P_2^{\alpha_2} \mbox{ is a first integral } \Leftrightarrow \alpha_1C_1 + \alpha_2C_2 = 0,
\end{equation}
and more generally
 \begin{equation}
\prod_{i}^{}P_i^{\alpha_i} \mbox{ is a first integral } \Leftrightarrow \sum_{i}^{}\alpha_iC_i  = 0.
\end{equation}
It follows that integrals that arise in this way are factorisable.

\medskip \noindent If there are more functionally independent DPs than the dimension of this linear space, then there must be one or more integrals.
The method we introduce here, produces $m$ integrals for an $n$-dimensional homogeneous quadratic ODE, from $n+m$ Darboux polynomials.
\begin{itemize}
\item Find $n$ independent DPs for the ODE:
\begin{equation}\label{Pdot}
\dot{P}_i(\mathbf{x}) = P_i(\mathbf{x})C_i(\mathbf{x}).
\end{equation}
The $C_i$ will be linear. Defining $\mathbf{v}$ to be the vector with components $v_i:= \ln(P_i)$, $i=1,\dots,n$, the equation (\ref{Pdot}) can be written as
\begin{equation}\label{vdot}
    \dot{\mathbf{v}} = \mathbf{A}\mathbf{x}
\end{equation}
where $\mathbf{A}$ is some constant invertible matrix.
\item Find $m$ additional DPs for the ODE ($m \leq n-1$ is a necessary condition for the integrals to be independent). Defining $\mathbf{w}$ to be the vector with components $w_i:= \ln(P_i)$, $i=n+1,\dots,n+m$, we get
\begin{equation}\label{wdot}
\dot{\mathbf{w}} = \mathbf{B}\mathbf{x}
\end{equation}
Eliminating $\mathbf{x}$, we again get
\begin{equation}\label{nDLV}
\dot{\mathbf{w}} - \mathbf{B}\mathbf{A}^{-1}\dot{\mathbf{v}} = 0 \rightarrow \mathbf{w} - \mathbf{B}\mathbf{A}^{-1}\mathbf{v} = \mathbf{I}.
\end{equation}
\end{itemize}

For $n$-component Lotka-Volterra (LV) systems, $n$ Darboux polynomials are given by the components of the vector $\mathbf{x}$, and we set $\mathbf{v}=\mathbf{x}$. From (\ref{nDLV}), by exponentiation of the logarithmic integrals $\mathbf{I}$, we obtain $m$ integrals of the form
\[
P_{n+i}^{|\mathbf{A}|}\prod_{j=1}^n x_j^{Z_{i,j}},\qquad i=1,\ldots,m,
\]
where
\begin{equation}\label{defZ}
\mathbf{Z}:=-\mathbf{B}\mathbf{A}^{-1}|\mathbf{A}|
\end{equation}
and $|\mathbf{A}|$ is the determinant of $\mathbf{A}$.

\section{Additional Darboux polynomials for Lotka-Volterra systems}

The complement of $\{i,k\}$ is denoted $\{i,k\}^{\rm c}:=\{1,2,\ldots,n\}\setminus\{i,k\}$.
\begin{lemma} \label{Leminhom}
Consider a system with
\begin{equation} \label{eik}
\dot{x}_i= x_i \left( b_i + \sum_{j=1}^nA_{i,j}x_j \right),\qquad
\dot{x}_k= x_k \left( b_k + \sum_{j=1}^nA_{k,j}x_j \right).
\end{equation}
The expression, with $\alpha\beta\neq0$,
\begin{equation} \label{DPadd}
P_{i,k} = \alpha x_i + \beta x_k,
\end{equation}
is a  DP if and only if, for some constant $b$ and all $j\in \{i,k \}^c$,
\begin{eqnarray}
A_{i,j} &=& A_{k,j}  \label{cond1} \\
b_i &=& b_k\ =\ b \label{cond2}\\
\alpha (A_{k,k}-A_{i,k}) &=& \beta (A_{k,i}-A_{i,i})\label{cond3}
\end{eqnarray}
and $(A_{k,k}-A_{i,k})(A_{k,i}-A_{i,i})\neq 0$.
\end{lemma}
\begin{proof}
We first show that the conditions (\ref{cond1}), (\ref{cond2}) and (\ref{cond3}) are sufficient, i.e. if they are satisfied, then $P_{i,k}$ defined by (\ref{DPadd}) is a DP for the ODE defined by (\ref{eik}). Equation (\ref{DPadd}) implies with (\ref{eik}) that
\begin{eqnarray*} \label{imps}
\alpha \dot{x}_i + \beta \dot{x}_k &=& \alpha x_i \left(b_i + \sum_{j=1}^nA_{i,j}x_j  \right) + \beta x_k \left(b_k +  \sum_{j=1}^nA_{k,j}x_j \right) \\
&=& \alpha x_i \left(b_i + A_{i,i}x_i + A_{i,k}x_k + \Sigma'  \right) + \beta x_k \left(b_k + A_{k,i}x_i + A_{k,k}x_k + \Sigma'  \right) \\
&=& (\alpha x_i + \beta x_k)b + \alpha A_{i,i}x_i^2 + (\alpha A_{i,k}+\beta A_{k,i})x_ix_k + \beta A_{k,k}x_k^2 + (\alpha x_i + \beta x_k)\Sigma' \\
& &\mbox{ using }(\ref{cond2}) \\
&=& (\alpha x_i + \beta x_k)(b + A_{i,i}x_i + A_{k,k}x_k + \Sigma') \mbox{ using }(\ref{cond3}),
\end{eqnarray*}
and where (using (\ref{cond1}))
\begin{equation}\label{Sig'}
\Sigma'  := \sum_{j\in \{i,k\}^c}^{} A_{i,j}x_j =  \sum_{j\in \{i,k\}^c}^{} A_{k,j}x_j.
\end{equation}
Next we show the conditions are necessary, i.e. if $P_{i,k}$ defined by (\ref{DPadd}) is a DP for the ODE defined by (\ref{eik}) then (\ref{cond1}), (\ref{cond2}) and (\ref{cond3}) hold. Equation (\ref{DPadd}) implies with (\ref{eik}) that
\begin{equation} \label{imps2}
\alpha \dot{x}_i + \beta \dot{x}_k = \alpha x_i \left(b_i + \sum_{j=1}^nA_{i,j}x_j  \right) + \beta x_k \left(b_k +  \sum_{j=1}^nA_{k,j}x_j \right).
\end{equation}
First consider  all terms that contain $x_j$ on the r.h.s., where $j \in\{i,k\}^{\rm c}$:
\begin{equation}
\alpha x_i A_{i,j} x_j + \beta x_k A_{k,j} x_j.
\end{equation}
This must vanish if we substitute
\begin{equation}\label{subs1}
x_k = -\frac{\alpha}{\beta}x_i.
\end{equation}
We find $\alpha (A_{i,j} - A_{k,j}) x_i x_j = 0$ and hence
\begin{equation}
A_{i,j} = A_{k,j}
\end{equation}
for all $j\in \{i,k \}^c$.

\medskip \noindent Now consider all remaining terms that do not contain any $x_j$, with $j\in\{i,k\}^{\rm c}$, i.e.
\begin{equation}\label{noxj}
\alpha x_i (b_i + A_{i,i}x_i +  A_{i,k}x_k ) + \beta x_k (b_k +  A_{k,i}x_i +  A_{k,k}x_k ).
\end{equation}
Once again (\ref{noxj}) must vanish if we substitute (\ref{subs1}).
Hence
\begin{equation}
x_i(b_i-b_k) + x_i^2 \left[ A_{i,i} - (\frac{\alpha}{\beta}A_{i,k} + A_{k,i}) + \frac{\alpha}{\beta}A_{k,k} \right] = 0,  \nonumber
\end{equation}
which implies that
\begin{equation}
b_i=b_k = b, \mbox{ say}, \nonumber
\end{equation}
and
\begin{equation}
\frac{\alpha}{\beta} = \frac{A_{i,i}-A_{k,i}}{A_{i,k}-A_{k,k}}. \nonumber
\end{equation}
\end{proof}
Of course several low-dimensional instances of Lemma 1 have appeared in papers by various authors over the years, cf. e.g. a 2D instance in equation (3.2) of \cite{HCFGL}, a 3D instance in Proposition 1\#(3) of \cite{Cairo}, and a 4D instance in equation (12) of \cite{DP}.

\section{Superintegrable $n$-component Lotka-Volterra systems, $n=2,3,4$}

\subsection{$n=2$}
The system
\begin{equation} \label{LV2}
  \left\{ \begin{array}{c}
              \dot{x}_1=x_{{1}} \left(a_{{1}}x_{{1}}+b_{{1}}x_{{2}} \right) \\
              \dot{x}_2=x_{{2}} \left(c_{{1}}x_{{1}}+a_{{2}}x_{{2}} \right)
            \end{array} \right.
\end{equation}
admits the Darboux polynomials $x_1,x_2$, with cofactors $a_{{1}}x_{{1}}+b_{{1}}x_{{2}},a_{{2}}x_{{2}}+c_{{1}}x_{{1}}$, and the Darboux polynomial $\left( c_{{1}}-a_{{1}} \right) x_{{1}}+ \left( a_{{2}}-b_{{1}}
 \right) x_{{2}}$, with cofactor $a_{{1}}x_{{1}}+a_{{2}}x_{{2}}$. They give rise to matrices
 \begin{equation}\label{A2}
 \mathbf{A}=\begin{pmatrix}a_{{1}}&b_{{1}}\\ c_{{1}}
&a_{{2}}\end{pmatrix} \text{ and } \mathbf{B}=\begin{pmatrix}a_{{1}}&a_{{2}}\end{pmatrix},
\end{equation}
and hence to the integral
\[
I=  \left(  \left( c_{{1}}-a_{{1}} \right) x_{{1}}+ \left( a_{{2}}-b_{{1}
} \right) x_{{2}} \right) ^{a_{{1}}a_{{2}}-b_{{1}}c_{{1}}}{x_{{1}}}^{-
a_{{2}} \left( a_{{1}}-c_{{1}} \right) }{x_{{2}}}^{-a_{{1}} \left( a_{
{2}}-b_{{1}} \right) }.
\]
\subsection{$n=3$}
The system
\begin{equation} \label{LV3}
\left\{ \begin{array}{c}
         \dot{x}_1=x_{{1}} \left( a_{{1}}x_{{1}}+b_{{1}}x_{{2}}+b_{{2}}x_{{3}} \right)\\
         \dot{x}_2=x_{{2}} \left( a_{{2}}x_{{2}}+b_{{2}}x_{{3}}+c_{{1}}x_{{1}} \right)\\
         \dot{x}_3=x_{{3}} \left( a_{{3}}x_{{3}}+c_{{1}}x_{{1}}+c_{{2}}x_{{2}} \right)
        \end{array}\right.
\end{equation}
relates to matrix
\begin{equation}\label{A3}
\A=\begin{pmatrix}a_{{1}}&b_{{1}}&b_{{2}} \\
c_{{1}}&a_{{2}}&b_{{2}}\\
c_{{1}}&c_{{2}}&a_{{3}}
\end{pmatrix}.
\end{equation}
The following are 2 additional Darboux polynomials:
\[
P_{1,2}=\left( c_{{1}}-a_{{1}} \right) x_{{1}}+ \left( a_{{2}}-b_{{1}}
 \right) x_{{2}},\quad P_{2,3}=\left( c_{{2}}-a_{{2}} \right) x_{{2}}+ \left( a_{{3
}}-b_{{2}} \right) x_{{3}},
\]
with cofactors
\[
C_{1,2}=a_{{1}}x_{{1}}+a_{{2}}x_{{2}}+b_{{2}}x_{{3}},\quad C_{2,3}=c_{{1}}x_{{1}}+a_{{2}}x_{{2}}+a_{{3}}x_
{{3}}.
\]
Thus we have
\[
\B=\begin{pmatrix}
a_{{1}}&a_{{2}}&b_{{2}}\\
c_{{1}}&a_{{2}}&a_{{3}},
\end{pmatrix}
\]
and we find $2=n-1$ integrals
\begin{align*}
I_1 &= \left(  \left( c_{{1}}-a_{{1}} \right) x_{{1}}+ \left( a_{{2}}-b_{{1}
} \right) x_{{2}} \right)^{|\A|}{x_{{1}}}^{- \left( a_{{2}}a_{{3}
}-b_{{2}}c_{{2}} \right)  \left( a_{{1}}-c_{{1}} \right) }{x_{{2}}}^{-
 \left( a_{{2}}-b_{{1}} \right)  \left( a_{{1}}a_{{3}}-b_{{2}}c_{{1}}
 \right) }{x_{{3}}}^{b_{{2}} \left( a_{{2}}-b_{{1}} \right)  \left( a_
{{1}}-c_{{1}} \right) },\\
I_2&=\left(  \left( c_{{2}}-a_{{2}} \right) x_{{2}}+ \left( a_{{3}}-b_{{2}
} \right) x_{{3}} \right) ^{|\A|}{x_{{1}}}^{c_{{1}} \left( a_{{3}}
-b_{{2}} \right)  \left( a_{{2}}-c_{{2}} \right) }{x_{{2}}}^{- \left(
a_{{2}}-c_{{2}} \right)  \left( a_{{1}}a_{{3}}-b_{{2}}c_{{1}} \right)
}{x_{{3}}}^{- \left( a_{{3}}-b_{{2}} \right)  \left( a_{{1}}a_{{2}}-b_
{{1}}c_{{1}} \right) }.
\end{align*}
A special case of (\ref{LV3}), where $a_1=-c_1$, $c_2=-a_2=b_1$ and $a_3=-b_2$, is linearly equivalent to the Euler top, which has an extra integral, cf. Appendix A.
\subsection{$n=4$}
\subsubsection{} The matrix
\begin{equation} \label{A4}
\A=\begin{pmatrix}
a_{{1}}&b_{{1}}&b_{{2}}&b_{{3}}\\
c_{{1}}&a_{{2}}&b_{{2}}&b_{{3}}\\ c_{{1}}&c_{{2}}&a_{{3}}&b_{{3}}\\
c_{{1}}&c_{{2}}&c_{{3}}&a_{{4}}
\end{pmatrix}
\end{equation}
has the property that $A_{i,j}=A_{i+1,j}$ for all $i\in\{1,2,3\}$ and $j\in\{i,i+1\}^{\rm c}$. The associated Lotka-Volterra system is
\begin{equation}\label{TT4}
\left\{  \begin{array}{c}
        \dot{x}_1=x_1 (a_1 x_1 + b_1 x_2 + b_2 x_3 + b_3 x_4 ) \\
        \dot{x}_2=x_2 (c_1 x_1 + a_2 x_2 + b_2 x_3 + b_ 3x_4 ) \\
        \dot{x}_3=x_3 (c_1 x_1 + c_2 x_2 + a_3 x_3 + b_3 x_4 ) \\
        \dot{x}_4=x_4 (c_1 x_1 + c_2 x_2 + c_3 x_3 + a_4 x_4 )
\end{array}\right.
\end{equation}
The system (\ref{TT4}) has 7 Darboux polynomials. The obvious ones are $P_i=x_i$, $i=1,2,3,4$, with cofactors $C_i=\sum_{j=1}^n A_{i,j}x_j$. The other three, obtained from Lemma \ref{Leminhom}, are:
\begin{align*}
P_{1,2}&=\left( c_{{1}}-a_{{1}} \right) x_{{1}} + \left( a_{{2}}-b_{{1}} \right) x_{{2}},\\
P_{2,3}&=\left( c_{{2}}-a_{{2}} \right) x_{{2}}+ \left( a_{{3}}-b_{{2}}
 \right) x_{{3}},\\
P_{3,4}&=\left( c_{{3}}-a_{{3}} \right) x_{{3}}+ \left( a_{{4}}-b_{{3}}
 \right) x_{{4}},
\end{align*}
with cofactors
\begin{align*}
C_{1,2}&=a_{{1}}x_{{1}}+a_{{2}}x_{{2}}+b_{{2}}x_{{3}}+b_{{3}}x_{{4}},\\
C_{2,3}&=c_{{1}}x_{{1}}+a_{{2}}x_{{2}}+a_{{3}}x_{{3}}+b_{{3}}x_{{4}},\\
C_{3,4}&=c_{{1}}x_{{1}}+c_{{2}}x_{{2}}+a_{{3}}x_{{3}}+a_{{4}}x_{{4}}.
\end{align*}
The coefficient matrix from these cofactors is
\[
\B=\begin{pmatrix}
a_{{1}}&a_{{2}}&b_{{2}}&b_{{3}}\\ c_{{1}}&a_{{2}}&a_{{3}}&b_{{3}}\\ c_{{1}}&c_{{2}}&a_{{3}}&a_{{4}}
\end{pmatrix}.
\]
The rather general method, introduced in section \ref{rgm}, gives rise to the following $3=n-1$ functionally independent integrals:
\[
I_i=P_{i,i+1}^{|\A|}
{x_{{1}}}^{Z_{i,1}}
{x_{{2}}}^{Z_{i,2}}
{x_{{3}}}^{Z_{i,3}}
{x_{{4}}}^{Z_{i,4}},\qquad i=1,2,3,
\]
where $I_1$ is determined by
\begin{align*}
Z_{{1,1}}&=- \left( a_{{2}}a_{{3}}a_{{4}}-a_{{2}}b_{{3}}c_{{3}}-a_{{3}}
b_{{3}}c_{{2}}-a_{{4}}b_{{2}}c_{{2}}+b_{{2}}b_{{3}}c_{{2}}+b_{{3}}c_{{
2}}c_{{3}} \right)  \left( a_{{1}}-c_{{1}} \right),\\
Z_{{1,2}}&=- \left( a_{{2}}-b_{{1}} \right)  \left( a_{{1}}a_{{3}}a_{{4
}}-a_{{1}}b_{{3}}c_{{3}}-a_{{3}}b_{{3}}c_{{1}}-a_{{4}}b_{{2}}c_{{1}}+b
_{{2}}b_{{3}}c_{{1}}+b_{{3}}c_{{1}}c_{{3}} \right),\\
Z_{{1,3}}&= \left( a_{{4}}b_{{2}}-b_{{3}}c_{{3}} \right)  \left( a_{{2}
}-b_{{1}} \right)  \left( a_{{1}}-c_{{1}} \right),\\
Z_{{1,4}}&=b_{{3}} \left( a_{{3}}-b_{{2}} \right)  \left( a_{{2}}-b_{{1
}} \right)  \left( a_{{1}}-c_{{1}} \right),
\end{align*}
$I_2$ is determined by
\begin{align*}
Z_{{2,1}}&=c_{{1}} \left( a_{{4}}-b_{{3}} \right)  \left( a_{{3}}-b_{{2
}} \right)  \left( a_{{2}}-c_{{2}} \right),\\
Z_{{2,2}}&=- \left( a_{{2}}-c_{{2}} \right)  \left( a_{{1}}a_{{3}}a_{{4
}}-a_{{1}}b_{{3}}c_{{3}}-a_{{3}}b_{{3}}c_{{1}}-a_{{4}}b_{{2}}c_{{1}}+b
_{{2}}b_{{3}}c_{{1}}+b_{{3}}c_{{1}}c_{{3}} \right),\\
Z_{{2,3}}&=- \left( a_{{3}}-b_{{2}} \right)  \left( a_{{1}}a_{{2}}a_{{4
}}-a_{{1}}b_{{3}}c_{{2}}-a_{{2}}b_{{3}}c_{{1}}-a_{{4}}b_{{1}}c_{{1}}+b
_{{1}}b_{{3}}c_{{1}}+b_{{3}}c_{{1}}c_{{2}} \right),\\
Z_{{2,4}}&=b_{{3}} \left( a_{{3}}-b_{{2}} \right)  \left( a_{{2}}-c_{{2
}} \right)  \left( a_{{1}}-c_{{1}} \right),
\end{align*}
and $I_3$ is determined by
\begin{align*}
Z_{{3,1}}&=c_{{1}} \left( a_{{4}}-b_{{3}} \right)  \left( a_{{3}}-c_{{3
}} \right)  \left( a_{{2}}-c_{{2}} \right),\\
Z_{{3,2}}&= \left( a_{{4}}-b_{{3}} \right)  \left( a_{{3}}-c_{{3}}
 \right)  \left( c_{{2}}a_{{1}}-c_{{1}}b_{{1}} \right),\\
Z_{{3,3}}&=- \left( a_{{3}}-c_{{3}} \right)  \left( a_{{1}}a_{{2}}a_{{4
}}-a_{{1}}b_{{3}}c_{{2}}-a_{{2}}b_{{3}}c_{{1}}-a_{{4}}b_{{1}}c_{{1}}+b
_{{1}}b_{{3}}c_{{1}}+b_{{3}}c_{{1}}c_{{2}} \right),\\
Z_{{3,4}}&=- \left( a_{{4}}-b_{{3}} \right)  \left( a_{{1}}a_{{2}}a_{{3
}}-a_{{1}}b_{{2}}c_{{2}}-a_{{2}}b_{{2}}c_{{1}}-a_{{3}}b_{{1}}c_{{1}}+b
_{{1}}b_{{2}}c_{{1}}+b_{{2}}c_{{1}}c_{{2}} \right).
\end{align*}

\subsubsection{}
Next we consider the matrix
\[
\A=\begin{pmatrix}
a_{{1}}&b_{{1}}&b_{{2}}&b_{{3}}\\
c_{{1}}&a_{{2}}&b_{{2}}&b_{{3}}\\
c_{{1}}&c_{{2}}&a_{{3}}&b_{{3}}\\
c_{{1}}&c_{{3}}&b_{{2}}&a_{{4}}
\end{pmatrix}.
\]
It has the property that $A_{i,j}=A_{k,j}$ for all $(i,k)\in\{(1,2),(2,3),(2,4)\}$ and $j\in\{i,k\}^{\rm c}$. The corresponding Lotka-Volterra system reads
\begin{equation}\label{A4B}
\left\{
\begin{array}{c}
        \dot{x}_1=x_1 (a_1 x_1 + b_1 x_2 + b_2 x_3 + b_3 x_4 ) \\
        \dot{x}_2=x_2 (c_1 x_1 + a_2 x_2 + b_2 x_3 + b_3 x_4 ) \\
        \dot{x}_3=x_3 (c_1 x_1 + c_2 x_2 + a_3 x_3 + b_3 x_4 ) \\
       \dot{x}_4=x_4 (c_1 x_1 + c_3 x_2 + b_2 x_3 + a_4 x_4 )
\end{array}\right.
\end{equation}
The additional Darboux polynomials are
\begin{align*}
P_{1,2}&=\left( c_{{1}}-a_{{1}} \right) x_{{1}} + \left( a_{{2}}-b_{{1}} \right) x_{{2}},\\
P_{2,3}&=\left( c_{{2}}-a_{{2}} \right) x_{{2}}+ \left( a_{{3}}-b_{{2}}
 \right) x_{{3}},\\
P_{2,4}&=\left( c_{{3}}-a_{{2}} \right) x_{{2}}+ \left( a_{{4}}-b_{{3}}
 \right) x_{{4}},
\end{align*}
with cofactors
\begin{align*}
C_{1,2}&=a_{{1}}x_{{1}}+a_{{2}}x_{{2}}+b_{{2}}x_{{3}}+b_{{3}}x_{{4}},\\
C_{2,3}&=c_{{1}}x_{{1}}+a_{{2}}x_{{2}}+a_{{3}}x_{{3}}+b_{{3}}x_{{4}},\\
C_{3,4}&=c_{{1}}x_{{1}}+a_{{2}}x_{{2}}+b_{{2}}x_{{3}}+a_{{4}}x_{{4}}.
\end{align*}
The coefficient matrix from these cofactors is
\[
\B=\begin{pmatrix}
a_{1} & a_{2} & b_{2} & b_{3}\\
c_{1} & a_{2} & a_{3} & b_{3}\\
c_{1} & a_{2} & b_{2} & a_{4}
\end{pmatrix}.
\]
We label the special pairs of indices of rows of $\mathbf{A}$ as follows,
\begin{equation} \label{e123}
e_1=(1,2),\qquad e_2=(2,3),\qquad e_3=(2,4).
\end{equation}
The same label can be used to enumerate the functionally independent integrals,
\[
I_i=P_{e_i}^{|A|}
{x_{{1}}}^{Z_{i,1}}
{x_{{2}}}^{Z_{i,2}}
{x_{{3}}}^{Z_{i,3}}
{x_{{4}}}^{Z_{i,4}},\qquad i=1,2,3,
\]
where
\begin{align*}
Z_{{1,1}}&=-\left(a_{2} a_{3} a_{4}-a_{2} b_{2} b_{3}-a_{3} b_{3} c_{3}-a_{4} b_{2} c_{2}+b_{2} b_{3} c_{2}+b_{2} b_{3} c_{3}\right) \left(a_{1}-c_{1}\right)\\
Z_{{1,2}}&=-\left(a_{2}-b_{1}\right) \left(a_{1} a_{3} a_{4}-a_{1} b_{2} b_{3}-a_{3} b_{3} c_{1}-a_{4} b_{2} c_{1}+2 b_{2} b_{3} c_{1}\right)\\
Z_{{1,3}}&=b_{2} \left(a_{4}-b_{3}\right) \left(a_{2}-b_{1}\right) \left(a_{1}-c_{1}\right)\\
Z_{{1,4}}&=b_{3} \left(a_{3}-b_{2}\right) \left(a_{2}-b_{1}\right) \left(a_{1}-c_{1}\right),
\end{align*}
\begin{align*}
Z_{{2,1}}&=c_{1} \left(a_{4}-b_{3}\right) \left(a_{3}-b_{2}\right) \left(a_{2}-c_{2}\right)\\
Z_{{2,2}}&=-\left(a_{2}-c_{2}\right) \left(a_{1} a_{3} a_{4}-a_{1} b_{2} b_{3}-a_{3} b_{3} c_{1}-a_{4} b_{2} c_{1}+2 b_{2} b_{3} c_{1}\right)\\
Z_{{2,3}}&=-\left(a_{3}-b_{2}\right) \left(a_{1} a_{2} a_{4}-a_{1} b_{3} c_{3}-a_{2} b_{3} c_{1}-a_{4} b_{1} c_{1}+b_{1} b_{3} c_{1}+b_{3} c_{1} c_{3}\right)\\
Z_{{2,4}}&=b_{3} \left(a_{3}-b_{2}\right) \left(a_{2}-c_{2}\right) \left(a_{1}-c_{1}\right),
\end{align*}
and
\begin{align*}
Z_{{3,1}}&=c_{1} \left(a_{4}-b_{3}\right) \left(a_{3}-b_{2}\right) \left(a_{2}-c_{3}\right)\\
Z_{{3,2}}&=-\left(a_{2}-c_{3}\right) \left(a_{1} a_{3} a_{4}-a_{1} b_{2} b_{3}-a_{3} b_{3} c_{1}-a_{4} b_{2} c_{1}+2 b_{2} b_{3} c_{1}\right)\\
Z_{{3,3}}&=b_{2} \left(a_{4}-b_{3}\right) \left(a_{2}-c_{3}\right) \left(a_{1}-c_{1}\right)\\
Z_{{3,4}}&=-\left(a_{4}-b_{3}\right) \left(a_{1} a_{2} a_{3}-a_{1} b_{2} c_{2}-a_{2} b_{2} c_{1}-a_{3} b_{1} c_{1}+b_{1} b_{2} c_{1}+b_{2} c_{1} c_{2}\right).
\end{align*}
The special pairs of indices of rows of $\mathbf{A}$ can be interpreted as edges of a tree, which we will do in the next section.
\section{Connection to trees}
To each of the above $n$-component Lotka-Volterra systems we associate a free (unrooted) tree $T$ on $n$ vertices as follows. The tree has an edge between vertex $i$ and vertex $k$ if the condition that $A_{i,j}=A_{k,j}$ for all $j\in\{i,k\}^{\rm c}$ is satisfied. Thus, the systems \eqref{A2}, \eqref{A3}, \eqref{A4} and \eqref{A4B} relate to the trees depicted in Figure \ref{trees}.
\begin{figure}[h]
\begin{center}
\scalebox{.75}{
\begin{tikzpicture}
	\node[shape=circle,draw=black,line width=0.5mm] (1) at (-2,0) {1};
	\node[shape=circle,draw=black,line width=0.5mm] (2) at (-2,2) {2};
	\path [-,line width=0.5mm, ] (1) edge node[left] {$1$} (2);
	\node at (-2,-0.5)[below]{$\tau_1$};
	
	\node[shape=circle,draw=black,line width=0.5mm] (1) at (1,0) {1};
	\node[shape=circle,draw=black,line width=0.5mm] (2) at (1,2) {2};
	\node[shape=circle,draw=black,line width=0.5mm] (3) at (1,4) {3};
	\path [-,line width=0.5mm, ] (1) edge node[left] {$1$} (2);
	\path [-,line width=0.5mm, ] (2) edge node[left] {$2$} (3);
	\node at (1,-0.5)[below]{$\tau_2$};
	
	\node[shape=circle,draw=black,line width=0.5mm] (1) at (4,0) {1};
	\node[shape=circle,draw=black,line width=0.5mm] (2) at (4,2) {2};
	\node[shape=circle,draw=black,line width=0.5mm] (3) at (4,4) {3};
	\node[shape=circle,draw=black,line width=0.5mm] (4) at (4,6) {4};
	\path [-,line width=0.5mm, ] (1) edge node[left] {$1$} (2);
	\path [-,line width=0.5mm, ] (2) edge node[left] {$2$} (3);
	\path [-,line width=0.5mm, ] (3) edge node[left] {$3$} (4);
	\node at (4,-0.5)[below]{$\tau_3$};
	
	\node[shape=circle,draw=black,line width=0.5mm] (1) at (8,0) {1};
	\node[shape=circle,draw=black,line width=0.5mm] (2) at (8,2) {2};
	\node[shape=circle,draw=black,line width=0.5mm] (3) at (7,4) {3};
	\node[shape=circle,draw=black,line width=0.5mm] (4) at (9,4) {4};
	\path [-,line width=0.5mm, ] (1) edge node[left] {$1$} (2);
	\path [-,line width=0.5mm, ] (2) edge node[left] {$2$} (3);
	\path [-,line width=0.5mm, ] (2) edge node[right] {$3$} (4);
	\node at (8,-0.5)[below]{$\tau_4$};
\end{tikzpicture}
}
\caption{\label{trees} The trees connected to the Lotka-Volterra systems \eqref{A2}, \eqref{A3}, \eqref{A4} and \eqref{A4B} (from left to right).}
\end{center}
\end{figure}
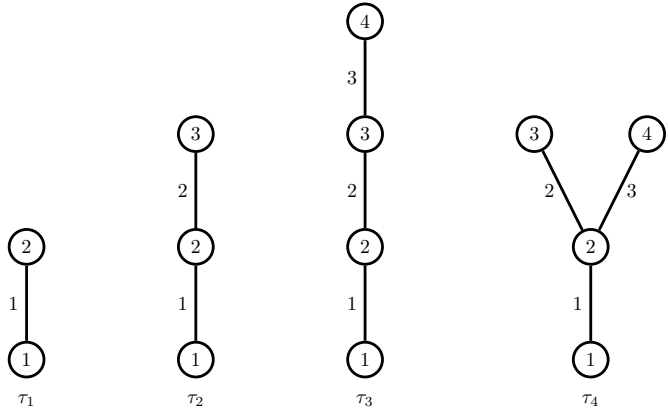

Vice versa, a tree $T$ on $n$ (ordered) vertices has $n-1$ (ordered) edges. We associated to $T$ a matrix $\A$ as follows. We start with an $n\times n$ diagonal matrix $\A$, with $A_{i,i}=a_i$. Then for each edge of $T$ we fix two off-diagonal entries of $\A$ as follows. For the $m$-th edge of the graph $T$, $e_m=(i,k)$ with $i<k$, we set $A_{i,k}=b_m$ and $A_{k,i}=c_m$. In \cite{KQM} we show that the remaining entries of the matrix $A$ are uniquely determined by the condition that $A_{i,j}=A_{k,j}$ when $(i,k)$ is an edge of $T$ and $j\in\{i,k\}^{\rm c}$. The matrix $\A$ has $3n-2$ free parameters and defines a Lotka-Volterra system
\begin{equation} \label{LVs}
\dot{x}_i=x_i\sum_{j=1}^nA_{i,j}x_j,\qquad i=1,2,\ldots,n,
\end{equation}
with $n-1$ integrals. In \cite{KQM}, we prove their functional independence, Theorem \ref{tfi}.
\begin{theorem} \label{tfi}
Each tree on $n$ vertices gives rise to a Lotka-Volterra system with $3n-2$ parameters, which admits $n-1$ functionally independent integrals.
\end{theorem}
One can think of the parameters $a_i$, $b_j$, $c_k$ as weights in a complete digraph $D$ (allowing both loops and multiple edges) which is associated to $T$. The matrix $A$ is then nothing but the adjacency matrix of $D$. The connection between Lotka-Volterra systems and graphs, via the adjacency matrix of the graph, has been made before \cite{Bog, DFO, Dam, EKV}, but in the context of undirected or directed graphs, and (mainly) anti-symmetric (and hence Hamiltonian) Lotka-Volterra systems. The general setting of complete digraphs seems to be new. Note that the number of trees is given by the sequence \cite[A000055]{OEIS}.

\section{Superintegrable 5-component Lotka-Volterra systems}
There are 3 non-isomorphic trees on 5 vertices, see Figure \ref{nitr}.
\begin{figure}[h]
\begin{center}
\scalebox{.75}{
	\begin{tikzpicture}				
		\node[shape=circle,draw=black,line width=0.5mm] (1) at (0,0) {1};
		\node[shape=circle,draw=black,line width=0.5mm] (2) at (0,2) {2};
		\node[shape=circle,draw=black,line width=0.5mm] (3) at (0,4) {3};
		\node[shape=circle,draw=black,line width=0.5mm] (4) at (0,6) {4};
		\node[shape=circle,draw=black,line width=0.5mm] (5) at (0,8) {5};		
		\path [-,line width=0.5mm, ] (1) edge node[left] {$1$} (2);
		\path [-,line width=0.5mm, ] (2) edge node[left] {$2$} (3);
		\path [-,line width=0.5mm, ] (3) edge node[left] {$3$} (4);
		\path [-,line width=0.5mm, ] (4) edge node[left] {$4$} (5);
		\node at (0,-0.5)[below]{$\tau_5$};
				
		\node[shape=circle,draw=black,line width=0.5mm] (1) at (4,0) {1};
		\node[shape=circle,draw=black,line width=0.5mm] (2) at (4,2) {2};
		\node[shape=circle,draw=black,line width=0.5mm] (3) at (4,4) {3};
		\node[shape=circle,draw=black,line width=0.5mm] (4) at (3,6) {4};
		\node[shape=circle,draw=black,line width=0.5mm] (5) at (5,6) {5};
		\path [-,line width=0.5mm, ] (1) edge node[left] {$1$} (2);
		\path [-,line width=0.5mm, ] (2) edge node[left] {$2$} (3);		
		\path [-,line width=0.5mm, ] (3) edge node[left] {$3$} (4);
		\path [-,line width=0.5mm, ] (3) edge node[left] {$4$} (5);
		\node at (4,-0.5)[below]{$\tau_6$};
		
		\node[shape=circle,draw=black,line width=0.5mm] (1) at (8,0) {1};
		\node[shape=circle,draw=black,line width=0.5mm] (2) at (8,2) {2};
		\node[shape=circle,draw=black,line width=0.5mm] (3) at (6.5,4) {3};
		\node[shape=circle,draw=black,line width=0.5mm] (4) at (9.5,4) {4};
		\node[shape=circle,draw=black,line width=0.5mm] (5) at (8,4) {5};
		\path [-,line width=0.5mm, ] (1) edge node[left] {$1$} (2);
		\path [-,line width=0.5mm, ] (2) edge node[left] {$2$} (3);
		\path [-,line width=0.5mm, ] (2) edge node[right] {$3$} (4);
		\path [-,line width=0.5mm, ] (2) edge node[right] {$4$} (5);
		\node at (8,-0.5)[below]{$\tau_7$};
	\end{tikzpicture}
}
\caption{\label{nitr} These are the three non-isomorphic trees on 5 vertices.}
\end{center}
\end{figure}
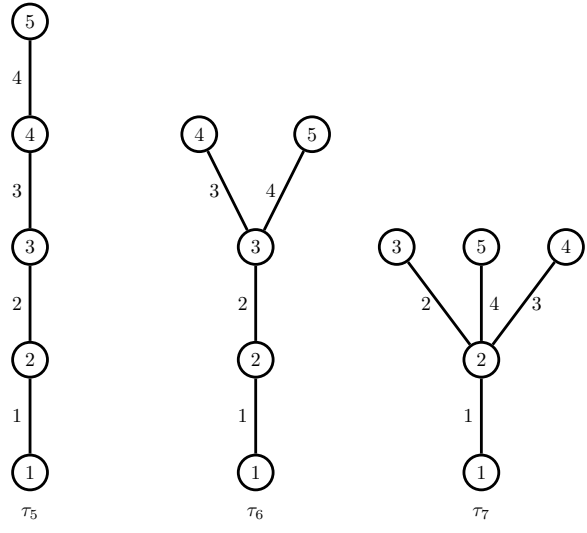
Following the procedure in the previous subsection, the trees in Figure \ref{nitr} give rise to matrices ($\mathbf{A}$)
\begin{equation} \label{AA}
\begin{pmatrix}
a_{1} & b_{1} & b_{2} & b_{3} & b_{4}
\\
 c_{1} & a_{2} & b_{2} & b_{3} & b_{4}
\\
 c_{1} & c_{2} & a_{3} & b_{3} & b_{4}
\\
 c_{1} & c_{2} & c_{3} & a_{4} & b_{4}
\\
 c_{1} & c_{2} & c_{3} & c_{4} & a_{5}
\end{pmatrix},\
\begin{pmatrix}
a_{{1}}&b_{{1}}&b_{{2}}&b_{{3}}&b_{{4}}\\
c_{{1}}&a_{{2}}&b_{{2}}&b_{{3}}&b_{{4}}\\
c_{{1}}&c_{{2}}&a_{{3}}&b_{{3}}&b_{{4}}\\
c_{{1}}&c_{{2}}&c_{{3}}&a_{{4}}&b_{{4}}\\
c_{{1}}&c_{{2}}&c_{{4}}&b_{{3}}&a_{{5}}
\end{pmatrix},\
\begin{pmatrix}
a_{{1}}&b_{{1}}&b_{{2}}&b_{{3}}&b_{{4}}\\
c_{{1}}&a_{{2}}&b_{{2}}&b_{{3}}&b_{{4}}\\
c_{{1}}&c_{{2}}&a_{{3}}&b_{{3}}&b_{{4}}\\
c_{{1}}&c_{{3}}&b_{{2}}&a_{{4}}&b_{{4}}\\
c_{{1}}&c_{{4}}&b_{{2}}&b_{{3}}&a_{{5}}
\end{pmatrix},
\end{equation}
and hence to Lotka-Volterra systems, each with 13 free parameters,
\begin{equation} \label{LV50}
\left\{
\begin{array}{c}
\dot{x}_1=x_{1} \left(a_{1} x_{1}+b_{1} x_{2}+b_{2} x_{3}+b_{3} x_{4}+b_{4} x_{5}\right)\\
\dot{x}_2=x_{2} \left(a_{2} x_{2}+b_{2} x_{3}+b_{3} x_{4}+b_{4} x_{5}+c_{1} x_{1}\right)\\
\dot{x}_3=x_{3} \left(a_{3} x_{3}+b_{3} x_{4}+b_{4} x_{5}+c_{1} x_{1}+c_{2} x_{2}\right)\\
\dot{x}_4=x_{4} \left(a_{4} x_{4}+b_{4} x_{5}+c_{1} x_{1}+c_{2} x_{2}+c_{3} x_{3}\right)\\
\dot{x}_5=x_{5} \left(a_{5} x_{5}+c_{1} x_{1}+c_{2} x_{2}+c_{3} x_{3}+c_{4} x_{4}\right),
\end{array}\right.
\end{equation}
\begin{equation} \label{LV51}
\left\{
\begin{array}{c}
\dot{x}_1=x_{{1}} \left( a_{{1}}x_{{1}}+b_{{1}}x_{{2}}+b_{{2}}x_{{3}}+b_{{3}}x_
{{4}}+b_{{4}}x_{{5}} \right)\\
\dot{x}_2=x_{{2}} \left( a_{{2}}x_{{2}}+b_{{2}}x_{
{3}}+b_{{3}}x_{{4}}+b_{{4}}x_{{5}}+c_{{1}}x_{{1}} \right)\\
\dot{x}_3=x_{{3}}
 \left( a_{{3}}x_{{3}}+b_{{3}}x_{{4}}+b_{{4}}x_{{5}}+c_{{1}}x_{{1}}+c_
{{2}}x_{{2}} \right)\\
\dot{x}_4=x_{{4}} \left( a_{{4}}x_{{4}}+b_{{4}}x_{{5}}+c_{
{1}}x_{{1}}+c_{{2}}x_{{2}}+c_{{3}}x_{{3}} \right)\\
\dot{x}_5=x_{{5}} \left( a_{{
5}}x_{{5}}+b_{{3}}x_{{4}}+c_{{1}}x_{{1}}+c_{{2}}x_{{2}}+c_{{4}}x_{{3}}
 \right)
\end{array}\right.
\end{equation}
and
\begin{equation} \label{LV52}
\left\{
\begin{array}{c}
\dot{x}_1=x_{{1}} \left( a_{{1}}x_{{1}}+b_{{1}}x_{{2}}+b_{{2}}x_{{3}}+b_{{3}}x_
{{4}}+b_{{4}}x_{{5}} \right)\\
\dot{x}_2=x_{{2}} \left( a_{{2}}x_{{2}}+b_{{2}}x_{
{3}}+b_{{3}}x_{{4}}+b_{{4}}x_{{5}}+c_{{1}}x_{{1}} \right)\\
\dot{x}_3=x_{{3}}
 \left( a_{{3}}x_{{3}}+b_{{3}}x_{{4}}+b_{{4}}x_{{5}}+c_{{1}}x_{{1}}+c_
{{2}}x_{{2}} \right)\\
\dot{x}_4=x_{{4}} \left( a_{{4}}x_{{4}}+b_{{2}}x_{{3}}+b_{
{4}}x_{{5}}+c_{{1}}x_{{1}}+c_{{3}}x_{{2}} \right)\\
\dot{x}_5=x_{{5}} \left( a_{{
5}}x_{{5}}+b_{{2}}x_{{3}}+b_{{3}}x_{{4}}+c_{{1}}x_{{1}}+c_{{4}}x_{{2}}
 \right).
\end{array}\right.
\end{equation}

Using the methods explained in sections 2 and 3, we can construct $4$ functionally independent integrals for each of these systems. As in section 4, the exponents in the integrals exhibit interesting factorisation properties. Below we provide the integrals for systems (\ref{LV50}), (\ref{LV51}) and (\ref{LV52}), expressing each exponent as a product of differences of parameters and a minor of $\mathbf{A}$. We let $\mathbf{A}^{I;J}$ denote the matrix $\mathbf{A}$ with rows $i\in I$ and columns $j\in J$ deleted. Its determinant $|\mathbf{A}^{I;J}|$ is called a minor of $\mathbf{A}$.

The Lotka-Volterra system (\ref{LV50}) admits the four functionally independent integrals
\begin{align*}
I_1&=\left(\left(c_{1}-a_{1}\right) x_{1}+\left(a_{2}-b_{1}\right) x_{2}\right)^{|\mathbf{A}|}x_1^{(c_{1}-a_{1})|\mathbf{A}^{1;1}|}x_2^{(b_{1}-a_{2})|\mathbf{A}^{2;2}|}
x_3^{\left(a_{2}-b_{1}\right)\left(a_{1}-c_{1}\right)|\mathbf{A}^{2,3;1,2}|}\\
 & \ \ \ \ \
x_4^{\left(a_{3}-b_{2}\right) \left(a_{2}-b_{1}\right) \left(a_{1}-c_{1}\right)|\mathbf{A}^{2,3,4;1,2,3}|}
x_5^{ \left(a_{4}-b_{3}\right) \left(a_{3}-b_{2}\right) \left(a_{2}-b_{1}\right) \left(a_{1}-c_{1}\right)b_{4}}\\
I_2&=\left(\left(c_{2}-a_{2}\right) x_{2}+\left(a_{3}-b_{2}\right) x_{3}\right)^{|\mathbf{A}|}
x_1^{\left(a_{5}-b_{4}\right) \left(a_{4}-b_{3}\right) \left(a_{3}-b_{2}\right) \left(a_{2}-c_{2}\right)c_{1}}
x_2^{(c_{2}-a_{2})|\mathbf{A}^{2;2}|}
x_3^{(b_{2}-a_{3})|\mathbf{A}^{3;3}|}\\
 & \ \ \ \ \
x_4^{\left(a_{3}-b_{2}\right) \left(a_{2}-c_{2}\right) \left(a_{1}-c_{1}\right)|\mathbf{A}^{1,3,4;1,2,3}|}x_5^{ \left(a_{4}-b_{3}\right) \left(a_{3}-b_{2}\right) \left(a_{2}-c_{2}\right) \left(a_{1}-c_{1}\right)b_{4}}\\
I_3&=\left(\left(c_{3}-a_{3}\right) x_{3}+\left(a_{4}-b_{3}\right) x_{4}\right)^{|\mathbf{A}|}x_1^{ \left(a_{5}-b_{4}\right) \left(a_{4}-b_{3}\right) \left(a_{3}-c_{3}\right) \left(a_{2}-c_{2}\right)c_{1} }x_2^{\left(a_{5}-b_{4}\right) \left(a_{4}-b_{3}\right) \left(a_{3}-c_{3}\right)|\mathbf{A}^{2,4,5;3,4,5}|}\\
 & \ \ \ \ \ x_3^{(c_{3}-a_{3})|\mathbf{A}^{3;3}|}x_4^{(b_{3}-a_{4})|\mathbf{A}^{4;4}|}x_5^{ \left(a_{4}-b_{3}\right) \left(a_{3}-c_{3}\right) \left(a_{2}-c_{2}\right) \left(a_{1}-c_{1}\right)b_{4}}\\
I_4&=\left(\left(c_{4}-a_{4}\right) x_{4}+\left(a_{5}-b_{4}\right) x_{5}\right)^{|\mathbf{A}|}
x_1^{\left(a_{5}-b_{4}\right) \left(a_{4}-c_{4}\right) \left(a_{3}-c_{3}\right) \left(a_{2}-c_{2}\right)c_{1} }
x_2^{\left(a_{5}-b_{4}\right) \left(a_{4}-c_{4}\right) \left(a_{3}-c_{3}\right)|\mathbf{A}^{2,3,5;3,4,5}|}\\
 & \ \ \ \ \
x_3^{\left(a_{5}-b_{4}\right) \left(a_{4}-c_{4}\right)|\mathbf{A}^{3,5;4,5}|}
x_4^{(c_{4}-a_{4})|\mathbf{A}^{4;4}|}
x_5^{(b_{4}-a_{5})|\mathbf{A}^{5;5}|}
\end{align*}
The Lotka-Volterra system (\ref{LV51}) admits the four functionally independent integrals
\begin{align*}
I_1&=\left(\left( c_{{1}}-a_{{1}} \right) x_{{1}}+ \left( a_{{2}}-b_{{1}}
 \right) x_{{2}}\right)^{|\A|}x_1^{(c_1-a_1)|\A^{1;1}|}x_2^{(b_1-a_2)|\A^{2;2}|}x_3^{(a_1-c_1)(a_2-b_1)|\A^{2,3;1,2}|}\\
 & \ \ \ \ \ x_4^{(a_1-c_1)(a_2-b_1)(a_3-b_2)(a_5-b_4)b_3}x_5^{(a_1-c_1)(a_2-b_1)(a_3-b_2)(a_4-b_3)b_4}\\
I_2&=\left( \left( c_{{2}}-a_{{2}} \right) x_{{2}}+ \left( a_{{3
}}-b_{{2}} \right) x_{{3}}\right)^{|\A|}x_1^{(a_2-c_2)(a_3-b_2)(a_4-b_3)(a_5-b_4)c_1}x_2^{(c_2-a_2)|\A^{2;2}|}x_3^{(b_2-a_3)|\A^{3,3}|}\\
 & \ \ \ \ \ x_4^{(a_1-c_1)(a_2-c_2)(a_3-b_2)(a_5-b_4)b_3}x_5^{(a_1-c_1)(a_2-c_2)(a_3-b_2)(a_4-b_3)b_4},\\
I_3&=\left(\left( c_{{3}}-a_{{3}} \right) x_{{3}}+
 \left( a_{{4}}-b_{{3}} \right) x_{{4}}\right)^{|\A|}x_1^{(a_2-c_2)(a_3-c_3)(a_4-b_3)(a_5-b_4)c_1}
 x_2^{(a_3-c_3)(a_4-b_3)(a_5-b_4)|\A^{2,4,5;3,4,5}|}\\
 & \ \ \ \ \  x_3^{(c_3-a_3)|\A^{3;3}|}x_4^{(b_3-a_4)|\A^{4;4}|}x_5^{(a_1-c_1)(a_2-c_2)(a_3-c_3)(a_4-b_3)b_4},\\
I_4&= \left( \left( c_{{4}}-a_{{3}}
 \right) x_{{3}}+ \left( a_{{5}}-b_{{4}} \right) x_{{5}}\right)^{|\A|}x_1^{(a_2-c_2)(a_3-c_4)(a_4-b_3)(a_5-b_4)c_1}
 x_2^{(a_3-c_4)(a_4-b_3)(a_5-b_4)|\A^{2,4,5;3,4,5}|}\\
 & \ \ \ \ \  x_3^{(c_4-a_3)|\A^{3;3}|}x_4^{(a_1-c_1)(a_2-c_2)(a_3-c_4)(a_5-b_4)b_3}x_5^{(b_4-a_5)|\A^{5;5}|}.
 \end{align*}
The Lotka-Volterra system (\ref{LV52}) admits the four functionally independent integrals
\begin{align*}
I_1&=\left( \left( c_{{1}}-a_{{1}} \right) x_{{1}}+ \left( a_{{2}}-b_{{1}}
 \right) x_{{2}}\right)^{|\A|}x_1^{(c_1-a_1)|\A^{1;1}|}x_2^{(b_1-a_2)|\A^{2;2}|}x_3^{(a_1-c_1)(a_2-b_1)(a_4-b_3)(a_5-b_4)b_2}\\
 & \ \ \ \ \
 x_4^{(a_1-c_1)(a_2-b_1)(a_3-b_2)(a_5-b_4)b_3}
 x_5^{(a_1-c_1)(a_2-b_1)(a_3-b_2)(a_4-b_3)b_4}, \\
I_2&=\left(   \left( c_{{2}}-a_{{2}} \right) x_{{2}}+ \left( a_{{3
}}-b_{{2}} \right) x_{{3}}\right)^{|\A|} x_1^{(a_2-c_2)(a_3-b_2)(a_4-b_3)(a_5-b_4)c_1}
x_2^{(c_2-a_2)|\A^{2;2}|}
x_3^{(b_2-a_3)|\A^{3;3}|}\\
 & \ \ \ \ \
x_4^{(a_1-c_1)(a_2-c_2)(a_3-b_2)(a_5-b_4)b_3}
x_5^{(a_1-c_1)(a_2-c_2)(a_3-b_2)(a_4-b_3)b_4}, \\
I_3&=\left(   \left( c_{{3}}-a_{{2}} \right) x_{{2}}+
 \left( a_{{4}}-b_{{3}} \right) x_{{4}}\right)^{|\A|}
 x_1^{(a_2-c_3)(a_3-b_2)(a_4-b_3)(a_5-b_4)c_1}
 x_2^{(c_3-a_2)|\A^{2;2}|}\\
 & \ \ \ \ \
 x_3^{(a_1-c_1)(a_2-c_3)(a_4-b_3)(a_5-b_4)b_2}
 x_4^{(b_3-a_4)|\A^{4;4}|}
 x_5^{(a_1-c_1)(a_2-c_3)(a_3-b_2)(a_4-b_3)b_4},\\
I_4&=\left(   \left( c_{{4}}-a_{{2}}
 \right) x_{{2}}+ \left( a_{{5}}-b_{{4}} \right) x_{{5}}\right)^{|\A|}
 x_1^{(a_2-c_4)(a_3-b_2)(a_4-b_3)(a_5-b_4)c_1}
 x_2^{(c_4-a_2)|\A^{2;2}|}\\
 & \ \ \ \ \
 x_3^{(a_1-c_1)(a_2-c_4)(a_4-b_3)(a_5-b_4)b_2}
 x_4^{(a_1-c_1)(a_2-c_4)(a_3-b_2)(a_5-b_4)b_3}
 x_5^{(b_4-a_5)|\A^{5;5}|}.
\end{align*}
The factorisation for the general case will be described in more detail in \cite{KQM}.

\section{A hierarchy of superintegrable Lotka-Volterra systems}
Consider the tall tree on $n$ vertices depicted in Figure \ref{ttree}.

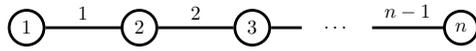
\begin{figure}[h]
\begin{center}
\scalebox{.75}{
	\begin{tikzpicture}				
	\node[shape=circle,draw=black,line width=0.5mm] (1) at (0,0) {1};
	\node[shape=circle,draw=black,line width=0.5mm] (2) at (2,0) {2};
	\node[shape=circle,draw=black,line width=0.5mm] (3) at (4,0) {3};
	\node[] (4) at (5,0) {};
	\node[] (5) at (5.5,0) {$\dots$};
	\node[] (6) at (6,0) {};
	\node[shape=circle,draw=black,line width=0.5mm] (7) at (7.7,0) {$n$};		
	\path [-,line width=0.5mm, ] (1) edge node[above] {$1$} (2);
	\path [-,line width=0.5mm, ] (2) edge node[above] {$2$} (3);
	\path [-,line width=0.5mm, ] (3) edge node[above] {} (4);
	\path [-,line width=0.5mm, ] (6) edge node[above] {$n-1$} (7);
\end{tikzpicture}}
\caption{\label{ttree} Tall tree on $n$ vertices.}
\end{center}
\end{figure}

It gives rise to the $n\times n$ matrix:
\begin{equation} \label{An}
\A=\begin{pmatrix}
a_{{1}}&b_{{1}}&b_{{2}}&b_{{3}}&\cdots&b_{n-1}\\
c_{{1}}&a_{{2}}&b_{{2}}&b_{{3}}&\cdots&b_{n-1}\\ c_{{1}}&c_{{2}}&a_{{3}}&b_{{3}}&\cdots&b_{n-1}\\
c_{{1}}&c_{{2}}&c_{{3}}&a_{{4}}&\cdots&b_{n-1}\\
\vdots&\vdots&\vdots&\vdots&\ddots&\vdots\\
c_1&c_2&c_3&c_4&\cdots&a_n
\end{pmatrix},
\end{equation}
of which matrices (\ref{A2}),(\ref{A3}),(\ref{A4}), and the left matrix in (\ref{AA}), are special cases taking $n=2,3,4$ and $5$ respectively.

For arbitrary $n$, the tall tree provides us with the Lotka-Volterra system:
\begin{equation} \label{LVn}
\left\{ \begin{split}
\dot{x}_1&=x_{1} \left(a_{1} x_{1}+b_{1} x_{2}+b_{2} x_{3}+\cdots+b_{n-1} x_{n}\right)\\
\dot{x}_2&=x_{2} \left(c_{1} x_{1}+a_{2} x_{2}+b_{2} x_{3}+\cdots+b_{n-1} x_{n}\right)\\
\dot{x}_3&=x_{3} \left(c_{1} x_{1}+c_{2} x_{2}+a_{3} x_{3}+\cdots+b_{n-1} x_{n}\right)\\
&\ \ \vdots\\
\dot{x}_{n-1}&=x_{n-1} \left(c_{1} x_{1}+c_{2} x_{2}+\cdots+ a_{n-1}x_{n-1}+ b_{n-1} x_{n}\right)\\
\dot{x}_n&=x_{n} \left(c_{1} x_{1}+c_{2} x_{2}+\cdots+ c_{n-1}x_{n-1}+ a_{n} x_{n}\right),
\end{split} \right.
\end{equation}

The $n$ coordinates $x_i$, $i=1,\ldots,n$, are Darboux polynomials. The system (\ref{LVn}) admits $n-1$ additional Darboux polynomials of the form
\[
P_{i,i+1}=\left( c_{{i}}-a_{{i}} \right) x_{{i}} + \left( a_{{i+1}}-b_{{i}} \right) x_{{i+1}},\qquad i=1,\ldots,n-1,
\]
with cofactors
\[
C_{i,i+1}=c_{1}x_{1}+\cdots+c_{i-1}x_{i-1}+a_ix_i
+a_{i+1}x_{i+1}+b_{i+1}x_{i+2}+\cdots b_{n-1}x_n.
\]

Their coefficients can be organised into the following $(n-1)\times n$ matrix:
\begin{equation} \label{Bn}
\B=\begin{pmatrix}
a_{{1}}&a_{{2}}&b_{{2}}&b_{{3}}&\cdots&b_{n-2}&b_{n-1}\\
c_{{1}}&a_{{2}}&a_{{3}}&b_{{3}}&\cdots&b_{n-2}&b_{n-1}\\ c_{{1}}&c_{{2}}&a_{{3}}&a_{{4}}&\cdots&b_{n-2}&b_{n-1}\\
\vdots&\vdots&\vdots&\vdots&\ddots&\vdots&\vdots\\
c_1&c_2&c_3&c_4&\cdots&a_{n-1}&b_{n-1}\\
c_1&c_2&c_3&c_4&\cdots&a_{n-1}&a_n
\end{pmatrix},
\end{equation}

Using the matrices $\A$ and $\B$, and defining $\Z=-\B\A^{-1}|\A|$, we obtain $n-1$ integrals of the form
\[
I_{i}=P_{i,i+1}^{|\A|}\prod_{j=1}^n x_j^{Z_{i,j}},\qquad i=1,\ldots,n-1.
\]
One can show, cf. \cite{KQM}, that the exponents factorise and that the integrals $K_i$ are functionally independent (which implies superintegrability). Introducing the notation
\[
\N_j^n=\{k\in\N:j\leq k\leq n\},
\]
we find, for all $i\in \N_1^{n-1},j\in\N_1^n$,
\[
Z_{i,j}=\begin{cases}
(a_i-c_i)\prod_{j<k<i}(a_k-c_k)\prod_{i<k\leq n}(a_k-b_{k-1})|\A^{\N_{j}^{i-1}\cap\N_{i+1}^n;\N_{j+1}^n}|&j<i,\\
(c_i-a_i)|\A^{i;i}|&j=i,\\
(b_i-a_{i+1})|\A^{i+1;i+1}|&j=i+1,\\
(a_i-c_i)\prod_{1<k<i}(a_k-c_k)\prod_{i<k< j}(a_k-b_{k-1})|\A^{\N_{1}^{i-1}\cap\N_{i+1}^j;\N_{1}^{j-1}}|&j>i+1.
\end{cases}
\]
This formula provides a more efficient way to calculate the exponents in the integrals $I_i$ than using the definition of $\Z$, which involves matrix multiplication, inversion and taking the determinant of an $n\times n$ matrix.

The special case $a_i=0 \quad (i=1,\dots,n), b_i = -c_{i+1} \quad (i=1,\dots,n-1)$ was studied in \cite{KQV}.

\bigskip

\noindent
{\bf Concluding remark.}
In this paper we have studied superintegrable Lotka-Volterra systems without imposing any additional structure. We intend to investigate the role of measure-preservation and symplectic structure on Lotka-Volterra equations in future work.

\medskip \noindent {\bf Acknowledgement} GRWQ is grateful to Silvia Perez Cruz for alleviating the plague years and to Sydney Mathematical Research Institute (SMRI) for travel support.

\appendix

\section{The Euler top}
The Euler top, in the form
\begin{equation} \label{LVE}
\left\{ \begin{split}
\dot{x}_1&=a^2x_2x_3 \\
\dot{x}_2&=b^2x_1x_3 \\
\dot{x}_3&=c^2x_1x_2,
\end{split}\right.
\end{equation}
admits the 6 Darboux polynomials
\[
cx_1\pm ax_3,\qquad
cx_2\pm bx_3,\qquad
bx_1\pm ax_2.
\]
Hence, it is linearly equivalent to an LV system with 3 additional Darboux polynomials.
In terms of $\mathbf{y}=(cx_1 + ax_3,cx_2+ bx_3,bx_1+ ax_2)/2$, we have
\begin{equation} \label{ET}
\dot{y}_i=y_i\sum_{j=1}^nA_{i,j}y_j,\qquad i=1,2,3,
\qquad \mathbf{A}=
\begin{pmatrix}
-b & a & c \\
b & -a & c \\
b & a & -c
\end{pmatrix},
\end{equation}
which is a special case of (\ref{LV3}). We note that the corresponding graph, see Figure \ref{nat}, is not a tree.

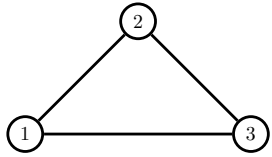
\begin{figure}[h]
\begin{center}
\scalebox{.75}{\begin{tikzpicture}
	\node[shape=circle,draw=black,line width=0.5mm] (1) at (0,0) {1};
	\node[shape=circle,draw=black,line width=0.5mm] (2) at (2,2) {2};
	\node[shape=circle,draw=black,line width=0.5mm] (3) at (4,0) {3};
	\path [-,line width=0.5mm, ] (1) edge node[above] {} (2);
    \path [-,line width=0.5mm, ] (2) edge node[above] {} (3);
    \path [-,line width=0.5mm, ] (1) edge node[above] {} (3);
\end{tikzpicture}}
\caption{\label{nat} Complete graph on $3$ vertices.}
\end{center}
\end{figure}

The additional Darboux polynomials are $ay_2-cy_3$, $by_1-cy_3$, and $ay_2-by_1$, and three integrals (not functionally independent) are given by
\[
y_1(ay_2-cy_3),\qquad y_3(by_1-ay_2),\qquad y_2(by_1-cy_3).
\]

It is now easy to generalise system (\ref{ET}) whilst keeping the same number of Darboux polynomials (six). Indeed, we would take
\[
\mathbf{A}=
\begin{pmatrix}
d & a & c \\
b & e & c \\
b & a & f
\end{pmatrix}.
\]
The corresponding LV system has additional Darboux polynomials
\[
P_1=\left( a-e \right) y_{{2}}+ \left( f-c \right) y_{{3}},\quad
P_2=\left( b-d \right) y_{{1}}+ \left( e-a \right) y_{{2}},\quad
P_3=\left( b-d \right) y_{{1}}+ \left( f-c \right) y_{{3}},
\]
cofactor coefficient matrix
\[
\B=\begin{pmatrix}
b & e & f \\
d & e & c \\
d & a & f
\end{pmatrix},
\]
and three integrals
\begin{align*}
{P_{{1}}}^{|\A|}{y_{{1}}}^{b \left( c-f \right)  \left( a-e \right) }{y_
{{2}}}^{- \left( bc-df \right)  \left( a-e \right) }{y_{{3}}}^{-
 \left( c-f \right)  \left( ab-de \right) },\\
 {P_{{2}}}^{|\A|}{y_{{1}}}^{-
 \left( b-d \right)  \left( ac-ef \right) }{y_{{2}}}^{- \left( bc-df
 \right)  \left( a-e \right) }{y_{{3}}}^{c \left( b-d \right)  \left(
a-e \right) },\\
{P_{{3}}}^{|\A|}{y_{{1}}}^{- \left( b-d \right)  \left( ac-
ef \right) }{y_{{2}}}^{a \left( c-f \right)  \left( b-d \right) }{y_{{
3}}}^{- \left( c-f \right)  \left( ab-de \right) },
\end{align*}
of which two are functionally independent.


\begin{thebibliography}{10}
\bibitem{BBM}
\'A. Ballesteros, A. Blasco and F. Musso, Integrable deformations of Lotka-Volterra systems, Phys. Lett. A 375 (2011) 3370-3374.

\bibitem{Bog}
O.I. Bogoyavlenskij, Integrable Lotka-Volterra Systems, Regular and Chaotic Dynamics 13(6) (2008) 543–556.

\bibitem{BV}
T. Bountis and P. Vanhaecke, Lotka-Volterra systems satisfying a strong Painlev\'e property, Phys. Lett. A. 380 (2016) 3977-3982.

\bibitem{BZDK}
T. Bountis, Z. Zhunussova, K. Dosmagulova, G. Kanellopoulos, Integrable and non-integrable Lotka-Volterra systems, Phys. Lett. A 402 (2021) 127360.

\bibitem{Cairo}
L. Cairo, Darboux first integral conditions and integrability of the 3D Lotka-Volterra system, J. Nonl. Math. Phys. 7(4) (2000) 511-531.

\bibitem{CD}
Y.T. Christodoulides, P.A. Damianou, Darboux polynomials for Lotka-Volterra systems in three dimensions, J. Nonl. Math. Phys. 16(3) (2009) 339-354.

\bibitem{Dam}
P.A. Damianou. Lotka-Volterra systems associated with graphs. In Group analysis of differential equations and integrable systems, pages 30–44. Department of Mathematics and Statistics, University of Cyprus, Nicosia, 2013.

\bibitem{DEKV}
P.A. Damianou, C.A. Evripidou, P. Kassotakis and P. Vanhaecke, Integrable reductions of the Bogoyavlenskij-Itoh Lotka-Volterra systems, J. Math. Phys. 58 (2017) 17pp.

\bibitem{DFO}
P. Duarte, R.L. Fernandes, W.M. Oliva, Dynamics of the attractor in the Lotka-Volterra equations, J. Differ. Equ. 149 (1998) 143-89.

\bibitem{DP}
P.A. Damianou, F. Petalidou, On the Liouville integrability of Lotka-Volterra systems, Front. Phys. 2 (2014) 1-10.

\bibitem{Darb}
G. Darboux, M\'emoire sur les \'equations diff\'erentielles alg\'ebraiques du premier ordre et du premier degr\'e, Bull. Sci Math. 1 (1878) 60--96, 123--144, 151--200.

\bibitem{EKV}
C.A. Evripidou, P. Kassotakis, P. Vanhaecke, Morphisms and automorphisms of skew-symmetric Lotka-Volterra systems, J. Phys. A: Math. and Theor. 55 (2022) 325201.

\bibitem{Gor}
A. Goriely, Integrability and Nonintegrability of Dynamical Systems (2001) World Scientific, section 2.5

\bibitem{HBF}
B. Hern\'andez-Bermejo and V. Fair\'en, Hamiltonian structure and Darboux theorem for families of generalized Lotka-Volterra systems, J. Math. Phys. 39 (1998) 6162-6174.

\bibitem{HCFGL}
D.D. Hua, L. Cair\'o, M.R. Feix, K.S. Govinder, P.G.L. Leach, Connection between the existence of first integrals and the Painlev\'e property in two-dimensional Lotka-Volterra and quadratic systems, Proc. R. Soc. Lond. A 452 (1996) 859-880.

\bibitem{Ito}
Y. Itoh, Integrals of a Lotka-Volterra system of odd number of variables, Progr. Theoret. Phys. 78 (1987) 507-510

\bibitem{KQV}
T.E. Kouloukas, G.R.W. Quispel and P. Vanhaecke, Liouville integrability and superintegrability of a generalized Lotka-Volterra system and its Kahan discretization, J. Phys. A 49 (2016) 13pp.

\bibitem{Lot}
A.J. Lotka, Analytical Theory of Biological Populations, The Plenum Series on Demographic Methods and Population Analysis, Plenum Press, New York, 1998. Translated from the 1939 French edition and with an introduction by David P. Smith and H\'el\`ene Rossert.


\bibitem{Pla}
M. Plank, Hamiltonian structures for the n-dimensional Lotka-Volterra equations, J Math Phys. 36 (1995) 3520–34.

\bibitem{PS}
M.J. Prelle and M.F. Singer,
Elementary first integrals of differential equations, Trans. AMS 279 (1983) 215-229.

\bibitem{OEIS}
N.J.A. Sloan, The online encyclopedia of integer sequences, https://oeis.org/A000055.

\bibitem{Ste}
S. Sternberg, Dynamical Systems, Dover, 2010.

\bibitem{KQM}
P.H. van der Kamp, G.R.W. Quispel, D.I. McLaren, Trees and superintegrable Lotka-Volterra families, submitted.

\bibitem{KKQTV}
P.H. van der Kamp, T. Kouloukas, G.R.W. Quispel, D. Tran and P. Vanhaecke, Integrable and superintegrable systems associated with multi-sums of products, Proc. R. Soc. Lond. Ser. A Math. Phys. Eng. Sci. 470 (2014) 201404.

\bibitem{Vol}
V. Volterra, Le\c{c}ons sur la th\'eorie math\'ematique de la lutte pour la vie, Les Grands Classiques Gauthier-Villars. Editions Jacques Gabay, Sceaux, 1990. Reprint of the 1931 original.
\end{thebibliography}
\end{document}